\documentclass[%
 reprint,
 amsmath,amssymb,
 aps,
]{revtex4-1}

\usepackage[dvipdfmx]{graphicx}
\usepackage{dcolumn}
\usepackage{bm}

\usepackage{amsthm}
\usepackage{color}

\newcommand{\balign}[1]{\begin{align} #1 \end{align}}
\newcommand{\fref}[1]{Fig.~\ref{fig:#1}}
\newcommand{\Fref}[1]{Figure~\ref{fig:#1}}
\newcommand{\average}[1]{\ensuremath{\langle#1\rangle} }
\newcommand{\eref}[1]{Eq.~\eqref{#1}}
\newcommand\mtr{T_{X\to Y}}
\newcommand\str{\overline{T}_{X\to Y}}
\newcommand{\del}{\partial}

\theoremstyle{plain}
\newtheorem{thm}{Theorem}

\newtheorem{cor}{Corollary}

\begin{document}

\preprint{APS/123-QED}

\title{Role of sufficient statistics in stochastic thermodynamics \\and its implication to sensory adaptation}

\author{Takumi Matsumoto}
\author{Takahiro Sagawa}%
\affiliation{%
 Department of Applied Physics, The University of Tokyo, 7-3-1 Hongo, Bunkyo-ku, Tokyo 113-8656, Japan
}%


\date{\today}

\begin{abstract}
A sufficient statistic is a significant concept in statistics, which means a probability variable that has sufficient information required for an inference task. We investigate the roles of sufficient statistics and related quantities in stochastic thermodynamics. Specifically, we prove that for general continuous-time bipartite networks, the existence of a sufficient statistic implies that an informational quantity called the sensory capacity takes the maximum. Since the maximal sensory capacity imposes a constraint that the energetic efficiency cannot exceed one-half, our result implies that the existence of a sufficient statistic is inevitably accompanied by energetic dissipation. We also show that, in a particular parameter region of linear Langevin systems, there exists the optimal noise intensity, at which the sensory capacity, the information-thermodynamic efficiency, and the total entropy production are optimized at the same time. We apply our general result to a model of sensory adaptation of {\it E. Coli}, and  find that the sensory capacity is nearly maximal with experimentally realistic parameters.
\end{abstract}

\pacs{Valid PACS appear here}
\maketitle



\section{\label{sec:int}Introduction}
In recent years, stochastic thermodynamics of small systems such as biomolecules have actively been studied both theoretically and experimentally~\cite{Seifert2012,Ciliberto2017}, where effects of thermal fluctuations are not negligible. The entropy of a small system can be reduced by feedback control based on measurements at the level of thermal fluctuations, which is a modern formulation of Maxwell's demon. Thermodynamics of information has rapidly been developed in this decade, and such a research direction is referred to as information thermodynamics~\cite{Parrondo2015}. 

In particular, thermodynamics of autonomous measurement and feedback has been developed~\cite{Allahverdyan2009,Sagawa2012,Ito2013,Hartich2014,Horowitz2014(2),Horowitz2014,Prokopenko2015,Shiraishi2015(2),Shiraishi2015,SMS2016,Yamamoto2016,Rosinberg2016,Hartich2016,phd} and applied to biological systems~\cite{Barato2013,phd,Barato2014,Sartori2014,Itoecoli,Sartori2015,Hartich2016,Ouldridge2017(PRL),Ouldridge2017(PRX)}, where the concept of  continuous information flow has played a significant role. Specifically, the transfer entropy~\cite{Allahverdyan2009,Ito2013,Sagawa2012,Prokopenko2015,Hartich2014,Hartich2016,phd,Schreiber2000} and the learning rate~\cite{Allahverdyan2009,Barato2014,Hartich2016,Brittain2017,Shiraishi2015(2),phd} have been shown related to the second law of thermodynamics. The ratio of these two informational quantities is referred to as the sensory capacity~\cite{Hartich2016,phd}, which is argued to be a measure of the effectiveness of stochastic sensors.  

On the other hand, a sufficient statistic is an important concept in statistics~\cite{Cover2006}, which means a probability variable that has sufficient information for an inference task. For example, the latest value of the estimator of the Kalman filter is a sufficient statistic~\cite{Astrom}. We note that the Kalman filter can be realized by a simple chemical reaction network~\cite{Kobayashi2010}.  The Kalman filter has also been studied in terms of information thermodynamics, where the sensory capacity is shown maximum~\cite{Horowitz2014,Hartich2016}.  However, the role of sufficient statistics in more general thermodynamic setups is yet to be understood.

In this paper, we investigate the connection between information thermodynamics, sufficient statistics, and some related concepts. We consider a general class of Markovian stochastic dynamics, described by the continuous-time bipartite network (CBN). In this setup, we prove that the existence of a sufficient statistic implies the maximum sensory capacity, which is a generalization of a previous result that the Kalman filter makes the sensory capacity maximum. This result implies that the sensory capacity is generally a reasonable measure of sufficiency of a statistic. It also follows that if there exists a sufficient statistic, then the information-thermodynamic efficiency cannot exceed one-half, given a trade-off relation between the efficiency and the sensory capacity~\cite{Hartich2016,phd}.

As a special case, we consider bipartite linear Langevin systems based on analytical solution. We show that under certain conditions, there exists the optimal noise intensity, at which the sensory capacity, the information-thermodynamic efficiency, and the total entropy production take the optimal values at the same time. 

We also consider sensory adaptation as an application of our theory to biological systems. In particular, we consider {\it E coli} chemotaxis, which can be modeled as a linear Langevin system~\cite{Barkai1997,Tu2008,Lan2012}. We compute the sensory capacity and the information-thermodynamic efficiency with experimentally realistic parameters. Our result shows that the sensory capacity is near optimal, while the information-thermodynamic efficiency           
does not even reach the half of the maximum as a consequence of the aforementioned trade-off relation. 

This paper is organized as follows. In Sec.~\ref{sec:set}, we describe our general setup and review fundamental concepts in information thermodynamics. In Sec.~\ref{sec:ss}, we prove that the existence of a sufficient statistic implies the maximum sensory capacity. In Sec.~\ref{sec:LLS}, we focus on bipartite linear Langevin systems, and derive the optimal noise intensity. In Sec.~\ref{sec:BC}, we consider a linear Langevin model of {\it E. coli} chemotaxis, and show that the sensory capacity becomes nearly maximum with realistic parameters. In Sec.~\ref{sec:conc}, we make concluding remarks. In Appendixes A and B, we give analytical expressions of the thermodynamic and informational quantities for linear Langevin systems. In Appendix C, we show examples of non-stationary dynamics of  a model of {\it E. Coli}.
\section{\label{sec:set}Setup}

In this section, we formulate our setup of Markovian stochastic dynamics with two subsystems, which is described by the CBN. As will be discussed in the following, the CBN applies to a broad class of stochastic systems including Langevin systems and Markovian jump systems. 

\begin{figure}[b]
\includegraphics[scale=0.3]{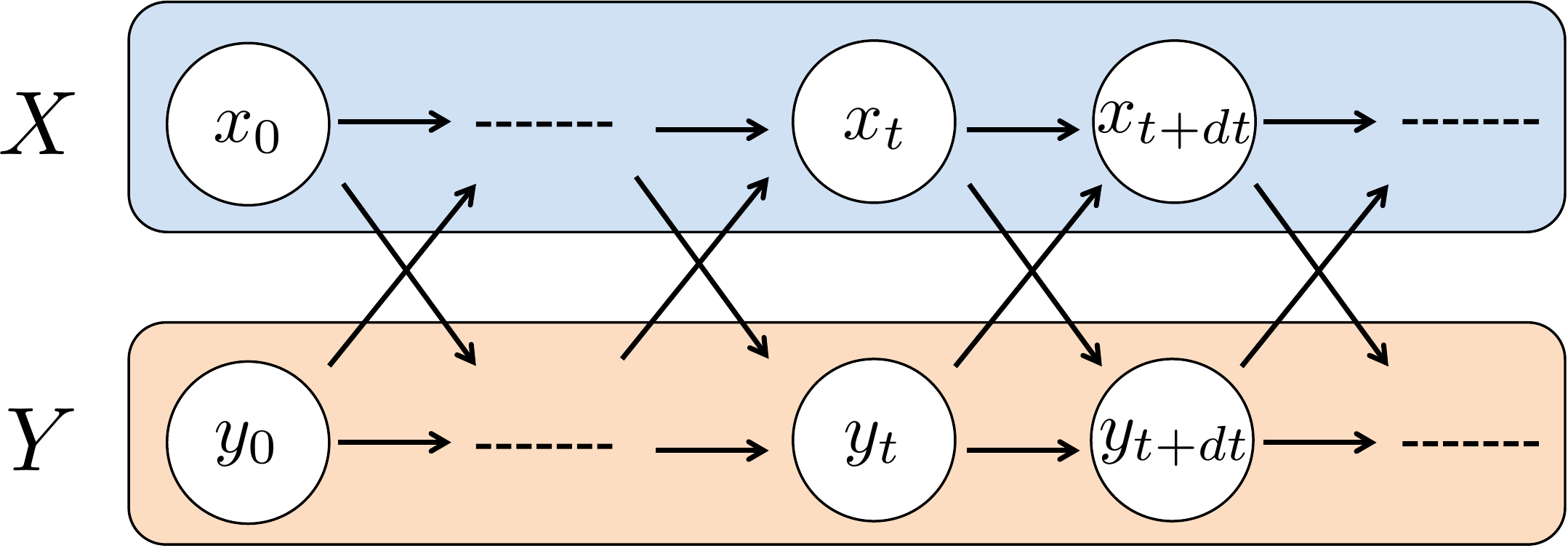}
\caption{\label{fig:dynamics}Graphical representation of the CBN. A node represents a state and an arrow represents a causal relationship.}
\end{figure}

We consider a stochastic system that consists of two subsystems, denoted as $X$, $Y$, whose microscopic states at time $t$ are respectively denoted by $x_t$ and $y_t$. We define the CBN as a Markovian stochastic process whose transition probability satisfies the bipartite property:
\balign{
p(x_{t+dt},y_{t+dt}|x_t,y_t)=p(x_{t+dt}|x_t,y_t)p(y_{t+dt}|x_t,y_t). \label{eq:bip}
}
\Fref{dynamics} graphically represents the causal relationship of the CBN. 

There are two important classes of the CBN. The first is coupled Langevin systems with independent noise:
\balign{
\begin{cases}
\dot{x}_t=f(x_t,y_t)+\xi^X_t, \\
\dot{y}_t=g(x_t,y_t)+\xi^Y_t,
\label{eq:ls}
\end{cases}
}
where $\xi^i_t\ (i=X,Y)$ are white Gaussian noise terms. They satisfy $\average{\xi^j_t\xi^j_{t'}}=2T^i\delta_{ij}\delta(t-t')$, which guarantees the bipartite property \eqref{eq:bip}. We note that such a Langevin system is often used for describing biochemical dynamics with the linear noise approximation~\cite{VanKampen,Hayot2004}. 

The second class of the CBN is coupled Markovian jump systems without simultaneous transition~\cite{Barato2013,Barato2013(2),Barato2014,phd,Hartich2014,Hartich2016}:
\balign{
w^{yy'}_{xx'}=\begin{cases}
w^{yy'}_x & (x=x', y\neq y'), \\
w^{y}_{xx'} & (x\neq x', y=y'), \\
0 & (x\neq x', y\neq y'),
\label{eq:rate}
\end{cases}
}
where $w^{yy'}_{xx'}$ is the transition rate from $(x,y)$ to $(x',y)$ with $x,x'$ and $y,y'$ respectively representing discrete states of $X$, $Y$. As seen from \eref{eq:rate}, the probability of simultaneous transition of the two systems is zero, which guarantees the bipartite property \eqref{eq:bip}. Such a Markovian jump process is also used for describing biochemical dynamics, where chemical states are coarse-grained to discrete variables~\cite{Barato2014}.

We next consider the informational quantities. In general, the strength of a correlation between two probability variables $X$, $Y$ is quantified by the mutual information~\cite{Cover2006}:
\balign{
I(X:Y)&:=\left<\ln \frac{p(x,y)}{p(x)p(y)}\right> \\
&=\sum_{x,y}p(x,y)\ln{\frac{p(x,y)}{p(x)p(y)}},
}
where $p(x,y)$ is the joint distribution, $p(x):=\sum_y p(x,y)$ and $p(y):=\sum_{x}p(x,y)$ are the marginal distributions, and $\langle \cdots \rangle$ denotes the ensemble average. Here, we used capital letters $X$, $Y$ for describing stochastic variables, while small letters $x$, $y$ for their particular realizations. If $X, Y$ are continuous variables, the sum in Eq.~(5) is replaced by the integral. 

In our CBN setup, $I(X_t, Y_t)$ characterizes the instantaneous correlation between the two systems at time $t$. Here again, capital letters $X_t$, $Y_t$ describe probability variables at time $t$, while small letters $x_t$, $y_t$ describe their particular realizations (i.e., particular microscopic states). As seen from the definition, the mutual information is symmetric between the two variables, and therefore, $I(X_t, Y_t)$ does not quantify the directional information flow from one to the other. 

To characterize the  directional information flow from one system to the other, we now consider the two informational quantities: the learning rate~\cite{Barato2014,Hartich2016,Brittain2017,Shiraishi2015(2),phd} and the transfer entropy~\cite{Ito2013,Sagawa2012,Prokopenko2015,Hartich2014,Hartich2016,phd}. The learning rate has recently been introduced in Ref.~\cite{Barato2014} in the context of information thermodynamics, while the transfer entropy has widely been used in a variety of fields such as time series analysis~\cite{time}. We note that the learning rate is also called just ``information flow"~\cite{Horowitz2014,Horowitz2014(2),Rosinberg2016}. 

First of all, the learning rate is defined as the partial derivative of the mutual information. The learning rate associated with $X$ during time $t$ to $t+dt$ is defined by
\balign{
l_X&:=\frac{I(X_{t+dt}:Y_t)-I(X_t:Y_t)}{dt}.
}
In the same manner, the learning rate associated with $Y$ is defined by 
\balign{
l_Y&:=\frac{I(X_t:Y_{t+dt})-I(X_t:Y_t)}{dt}. \label{eq:der}
}
These definitions immediately leads to
\balign{
\frac{d}{dt}I(X_t:Y_t)=l_X+l_Y. \label{eq:partial}
}

We note that the learning rate becomes either positive or negative. In particular, since in the stationary state the left-hand side of \eref{eq:partial} is zero, if the learning rate of a system is positive, then that of the other is negative. 

We consider the meaning of learning rate in more detail. When $l_{Y}>0$, $Y$ gains information about $X_t$, and thus $X_t$ has a larger amount of correlation with $Y_{t+dt}$ than with $Y_t$. In such a case, $Y$ plays the role of a {\it memory} that stores information about $X$. On the other hand, when $l_X<0$, the correlation between $X_t$ and $Y_t$ is utilized by $Y$ by feedback control, or is just dissipated into the environment~\cite{Horowitz2014}. In such a case, $X$ does not play the role of a memory, but is just referred to as a {\it system} that is controlled by the other subsystem. Especially in the stationary state, the subsystem with the positive learning rate is called a memory, and the subsystem with the negative learning rate is called a system. In the following, we suppose that $Y$ is a memory with $l_Y\geq0$. 

We next discuss the transfer entropy. In general, the conditional mutual information between probability variables $X$ and $Y$ under the condition of $Z$ is defined as~\cite{Cover2006}:
\balign{
I(X:Y|Z)&:=\left<\ln\frac{p(x,y|z)}{p(x|z)p(y|z)}\right> \\
&=\sum_{x,y,z}p(x,y,z)\ln{\frac{p(x,y|z)}{p(x|z)p(y|z)}}.
}
In the CBN, the transfer entropy from $X$ to $Y$ is defined by the following conditional mutual information:
\balign{
\mtr:=\frac{I(X_t:Y_{t+dt}|\left\{Y_{t'}\right\}_{t'\leq t})}{dt},
}
where $\left\{ Y_{t'} \right\} _{t'\leq t}$ represents a trajectory of $Y$ from time $t'$ to $t$. The transfer entropy quantifies the newly obtained information by $Y$ from time $t$ to $dt$.  In other words, the transfer entropy quantifies how strongly the value of $X_t$ has influenced dynamics of $Y$. 

In practice, however, it is generally hard to obtain the probabilities under the condition of the whole trajectory $\left\{ Y_{t'} \right\} _{t'\leq t}$ from numerical or experimental data.
We thus define a slightly different version of the transfer entropy, be computed more easily, by 
\balign{
\str:=\frac{1}{dt}I(X_t:Y_{t+dt}|Y_t),
}
where the mutual information is conditioned only by $Y_t$ at the latest time $t$.
We refer to $\mtr$ as the multi-time step transfer entropy (m-TE), and $\str$ as the single-time step transfer entropy (s-TE).

It is known that the s-TE gives an upper bound of the m-TE: $\str\geq \mtr$~\cite{Hartich2014}. Furthermore, the m-TE gives an upper bound  of the learning rate $\mtr\geq l_{Y}$~\cite{Hartich2014,Hartich2016}. Therefore, we have a hierarchy of informational quantities:
\balign{
\str\geq \mtr\geq l_Y. \label{eq:bigger} 
}

We now consider the sensory capacity~\cite{Hartich2016,phd}, which is defined by

\balign{
C_Y:= \frac{l_Y}{\mtr}. 
}
We refer to $C_Y$ as the multi-step sensory capacity (m-SC). In the case of $l_Y\geq0$, $0\leq C_Y\leq 1$ holds from inequality \eqref{eq:bigger}. In this sense, $C_Y$ is meaningful only when $l_Y \geq 0$.  In the case that $l_Y <0$ but $l_X \geq 0$, we instead consider $C_X$.  Throughout this paper, we focus on the situation of $l_Y \geq 0$. 

We can also define the single-step sensory capacity (s-SC) by using the s-TE~\cite{phd}:
\balign{
\overline{C}_Y:=\frac{l_Y}{\str}. \label{eq:Dc}
}
In the case of $l_Y\geq0$, $0\leq \overline{C}_Y\leq 1$ holds from inequality \eqref{eq:bigger}. In addition, inequality \eqref{eq:bigger} implies that $\overline{C}_Y\leq C_Y$. These sensory capacities characterize a kind of effectiveness of information gain by $Y$.  We will investigate the information-theoretic meaning of $C_Y$ and $\overline{C}_Y$ in more detail in the next section. 

We next consider the entropy production and the energetic efficiencies. The entropy production on the CBN is defined in terms of the ratio of the forward and backward path probabilities of stochastic dynamics. In the case of the Langevin system \eqref{eq:ls}, the entropy production in $Y$ is defined by~\cite{Sekimoto}
\balign{
\sigma_Y:=\frac{1}{dt}\left< \ln \frac{p(y_{t+dt})p(y_{t+dt}|x_t,y_t)}{p(y_t)p_B(y_t|x_t,y_{t+dt})} \right>,
}
where $p_B(y_t|x_t,y_{t+dt})$ is the backward path probability defined as
\balign{
p_B(y_t|x_t,y_{t+dt}):=N\exp\left[-\frac{(y_{t+dt}-y_t-g(x_t,y_{t+dt})dt)^2}{4T^Ydt}\right]
}
with $N$ being a normalization factor. In the case of the Markovian jump process \eqref{eq:rate}, the entropy production of $Y$ is defined by~\cite{Seifert2012}
\balign{
\sigma_Y:=\sum_{x,y,y'}p(x,y)w^{x}_{yy'}\ln\frac{w^{x}_{yy'}}{w^{x}_{y'y}}.
}
We can also define $\sigma_{X}$ in the same manner.  The total entropy production is defined by $\sigma_{\rm tot} := \sigma_{X} + \sigma_{Y}$.

The information-thermodynamic efficiency of $Y$ is defined as the ratio between the learning rate and the entropy production:
\balign{
\eta_{Y}:=\frac{l_Y}{\sigma_Y}. \label{eq:Et}
}
Since the second law of information thermodynamics is represented as $\sigma_Y\geq l_{Y}$~\cite{Shiraishi2015(2),Horowitz2014(2)}, the information-thermodynamic efficiency satisfies 
\balign{
\eta_{Y}\leq 1. \label{eq:t2nd}
}

In biochemical systems, the entropy production represents the free-energy dissipation by ATP hydrolysis. Therefore, we can regard the information-thermodynamic efficiency \eqref{eq:Et} as the energetic efficiency, which characterizes how efficiently the memory uses the free energy to get information. 

Interestingly, the sensory capacity and the information-thermodynamic efficiency are quantitatively related with each other. It has been shown in Ref.~\cite{Hartich2016} that if the m-SC takes the maximum (i.e., $C_Y=1$), then the information-thermodynamic efficiency must be less than or equal to the half of the maximum value:
\balign{
\eta_{Y}\leq1/2. \label{eq:half}
}
This implies that the effectiveness of information gain puts a constraint on the energetic efficiency. We note that inequality \eqref{eq:half} has originally been proved for stationary states without feedback~\cite{Hartich2016}, while its generalization to general situations is straightforward.
\section{\label{sec:ss}Sufficient statistics}
In this section, we show that the foregoing informational quantities (i.e., the two transfer entropies and the learning rate) are all equivalent in the presence of a sufficient statistic, implying that these informational quantities provide a good indicator of a sufficient statistic. This is one of the main results of this paper. 

A sufficient statistic enables a kind of optimal inference~\cite{Cover2006}.  In the present context, we focus on the optimality of inference about $X$ by $Y$ in stochastic dynamics.  In particular, we consider the situation that the latest state of $Y$ is a sufficient statistic of the latest state of $X$, where the entire past trajectory of $Y$ is not necessary to infer the latest state of $X$. This is indeed the case for the Kalman filter of a linear Langevin system~\cite{Astrom}, where the Bayesian estimator $Y$ is a sufficient statistic of $X$.
We note that the following argument is not restricted to the Kalman filter, but applies to any CBN.

We set the initial time of dynamics to $t=0$. When the probability distribution on the CBN has the following property, we say that $y_t$ is a sufficient statistic of $x_t$:
\balign{
p(x_t|\left\{y_{t'}\right\}_{0\leq t'\leq t})=p(x_t|y_t). \label{eq:Mark}
}
This means that the entire trajectory $\left\{ y_{t'} \right\}_{0\leq t' \leq t}$ does not have more information about $x_t$ than that $y_t$ has. In other words, $y_t$ provides an optimal inference of $x_t$. We emphasize that \eref{eq:Mark} is different from the Markovian property of the total system $(x_t, y_t)$.  

In the following, we discretize time by $t= k \Delta t$ with $\Delta t > 0$ and $k = 0, 1, 2, \cdots$. With this notation, \eref{eq:Mark} is written as 
\balign{
p(x_k | y_0^k) = p(x_k | y_k) \label{eq:disc},
}
where $x_k := x(k \Delta t)$, $y_k := y(k \Delta t)$, and $y_0^k := (y_0, y_1, \cdots, y_k)$. We will take the limit of $\Delta t \to 0$ at the final stage.

We first note that if \eref{eq:disc} is satisfied, the dynamics of the memory alone has the  Markovian property:  
\balign{
p(y_{k+1}|y^k_0)=p(y_{k+1}|y_k) \label{eq:prf2},
}
because we  have
\balign{
p(y_{k+1}|y_0^k)&=\sum_{x_k}p(x_k,y_{k+1}|y_0^k) \notag\\
&=\sum_{x_k} dx_k p(y_{k+1}|x_k,y_0^k)p(x_k|y_0^k) \notag\\
&=\sum_{x_k} p(y_{k+1}|x_k,y_k)p(x_k|y_k) \notag\\
&=\sum_{x_k} p(x_k,y_{k+1}|y_k) \notag\\
&=p(y_{k+1}|y_k). \label{eq:prf1}
}
Here, to obtain the third equality, we used that $(x_k,y_k)$ obeys the bipartite Markovian dynamics  and that $y_k$ is a sufficient statistic of $x_k$.
In the case of the Kalman filter, \eref{eq:prf1} implies that the estimator obeys the innovation process~\cite{innovation}.

We now show the equivalence of the m-TE and the s-TE in the presence of sufficient statistic.
\begin{thm}[Equivalence of the transfer entropies]
{\rm If $y_k$ is a sufficient statistic of $x_k$, we have} 
\balign{
\mtr=\str.
}
{\rm We note that this is true even without taking the limit of $\Delta t \to 0$.}
\end{thm}
\begin{proof}
{\rm It is straightforward to show that}
\balign{
\mtr&=\frac{1}{\Delta t}\left<\ln \frac{p(x_k,y_{k+1}|y_0^k)}{p(x_k|y_0^k)p(y_{k+1}|y_0^k)}\right> \notag\\
&=\frac{1}{\Delta t}\left< \ln\frac{p(y_{k+1}|x_k,y_0^k)}{p(y_{k+1}|y_0^k)} \right> \notag\\
&=\frac{1}{\Delta t}\left< \ln \frac{p(y_{k+1}|x_k,y_k)}{p(y_{k+1}|y_k)} \right> \notag\\
&=\frac{1}{\Delta t}\left< \ln\frac{p(x_k,y_{k+1}|y_k)}{p(x_k|y_k)p(y_{k+1}|y_k)}\right> \notag\\
&=\str,
}
{\rm where we used \eref{eq:prf2} to obtain the third equality.}
\end{proof}

We next discuss the equivalence of the s-TE and the learning rate in the presence of a sufficient statistic.
\begin{thm}[Equivalence of the s-TE and the learning rate]
{\rm If $y_t$ is a sufficient statistic of $x_t$, we have}
\balign{
\str=l_{Y}, \label{eq:claim2}
}
{\rm or equivalently, the s-SC takes the maximum (i.e., $\overline{C}_Y=1$). We note that this is true only in the limit of $\Delta t \to 0$.}
\end{thm}
\begin{proof}
We first note that 
\balign{
l_Y\Delta t=&-\left\{I(X_{k+1}:X_k)-I(X_k:Y_k)\right\} \notag\\
&+I(X_{k+1}:Y_{k+1})-I(X_k:Y_k)+\mathcal{O}(\Delta t^2).
}
We then consider the following quantity:
\balign{
&\str\Delta t-l_Y\Delta t-I(X_{k+1}:Y_k|Y_{k+1}) \notag\\
=&I(X_k:Y_{k+1}|Y_k)-I(X_{k+1}:Y_k|Y_{k+1})+\mathcal{O}(\Delta t^2) \notag\\
 &-I(X_{k+1}:Y_{k+1})+I(X_{k+1}:Y_k) \notag\\
=&I(X_k:Y_{k+1}|Y_k)-\left\{I(Y_{k+1}:Y_{k+1},Y_k)-I(X_{k+1}:Y_{k+1})\right\} \notag\\
& -I(X_{k+1}:Y_{k+1})+I(X_{k+1}:Y_k)+\mathcal{O}(\Delta t^2)\notag\\
=&I(X_k:Y_{k+1}|Y_k)-I(X_{k+1}:Y_{k+1},Y_k) \notag \\
&+I(X_{k+1}:Y_k)+\mathcal{O}(\Delta t^2)\notag\\
=&I(X_k:Y_{k+1}|Y_k)-I(X_{k+1}:Y_{k+1}|Y_k)+\mathcal{O}(\Delta t^2). 
}
Here, the last line satisfies that~\cite{Ito2013,Ito2016,Hartich2014}
\balign{
I(X_k:Y_{k+1}|Y_k)-I(X_{k+1}:Y_{k+1}|Y_k)=\mathcal{O}(\Delta t^2), \label{eq:point2}
}
and therefore
\balign{
\str-l_{Y}=\frac{1}{\Delta t}I(X_{k+1}:Y_k|Y_{k+1})+\mathcal{O}(\Delta t).  \label{eq:point}
}
On the other hand, we have
\balign{
I(X_{k+1}:Y_k|Y_{k+1})&=\left< \ln\frac{p(x_{k+1},y_k|y_{k+1})}{p(x_{k+1}|y_{k+1})p(y_k|y_{k+1})} \right> \notag \\
&=\left< \ln \frac{p(x_{k+1}|y_{k+1},y_k)}{p(x_{k+1}|y_{k+1})} \right> \notag \\
&=\left< \ln \frac{p(x_{k+1}|y_{k+1})}{p(x_{k+1}|y_{k+1})} \right> \notag \\
&=0, \label{eq:co}
}
where we used the definition of a sufficient statistic to obtain the third equality.
By combining Eqs.~\eqref{eq:point} and \eqref{eq:co}, we obtain \eref{eq:claim2}.
\end{proof}
As seen in \eref{eq:co}, it is sufficient to assume that $p(x_{k+1}|y_{k+1},y_k)=p(x_{k+1}|y_{k+1})$ to obtain \eref{eq:claim2}.  This is consistent with the fact that both of the s-TE and the learning rate only concern stochastic variables at time $k$ and $k+1$. 

By combining Theorem 1 and 2, we obtain the following.
\begin{cor}
{\rm If $y_t$ is a sufficient statistic of $x_t$, we have}
\balign{
\mtr=\str=l_{Y},
}
{\rm or equivalently, the m-SC takes the maximum (i.e., $C_Y=1$).
This is true only in the limit of $\Delta t \to 0$.}
\end{cor}

We summarize the foregoing results as follows:\\

\noindent
(i) $p(x_t|\left\{y_{t'}\right\}_{t'\leq t})=p(x_t|y_t)$
\balign{
\Rightarrow \mtr=\str=l_{Y},\ C_Y=1. \label{eq:mtrstr}
}
(ii) $p(x_{t+dt}|y_{t+dt},y_t)=p(x_{t+dt}|y_{t+dt})$
\balign{
\Rightarrow \str=l_{Y}, \ \overline{C}_Y=1. \label{eq:mtrstr2}\ \ \ \ \ \ \ \ \ \ \ \ 
}
These results imply that the sensory capacities ($C_Y$ and $\overline{C}_Y$) are good indicators of a sufficient statistic.
Therefore, we can adopt $C_Y$ and $\overline{C}_Y$ to quantify how a stochastic variable is close to a sufficient statistic. Combining our result \eqref{eq:mtrstr} and a previous result \eqref{eq:half}, we find that the existence of a sufficient statistic inevitably leads to energetic dissipation.

As mentioned before, the Kalman filter of a linear Langevin system provides a sufficient statistic~\cite{Astrom}, where $\mtr=l_{Y}$ has been proved in Ref.~\cite{Horowitz2014}. On the other hand, our argument here does not rely on the details of the system, and thus is not restricted to the Kalman filter of a linear system; our results are applicable to a broad class of stochastic dynamics described by the CBN, including nonlinear systems.

In Ref.~\cite{Hartich2016}, it has been proved that $p(x_t|\left\{y_{t'}\right\}_{t'\leq t})=p(x_t|y_t) \Leftrightarrow C_Y=1$ for stationary states. (While the original proof is for feed-forward systems, its generalization to feedback systems is straightforward, as long as the system is in a stationary state.)  However, $``\Leftarrow"$ is not generally true for non-stationary states, as seen from our proof above.
\section{\label{sec:LLS}Optimal Noise Intensity in Linear Langevin Systems}
In this section, we discuss a more quantitative connection between the entropy production and the informational quantities (and thus sufficient statistics), by focusing on linear Langevin systems as a special case of the CBN:
\begin{eqnarray}
\begin{cases}
\dot{x}_t=a^{11}x_t+a^{12}y_t+b^1+\xi^X_t, \\
\dot{y}_t=a^{21}x_t+a^{22}y_t+b^2+\xi^Y_t. \label{eq:general}
\end{cases}
\end{eqnarray}
Here, $\xi^i_t\ (i=X,Y)$ are the white Gaussian noise terms, which satisfy $\average{\xi^i_t}=0$ and
 $\average{\xi^i_t\xi^j_{t'}}=2T^i\delta_{ij}\delta(t-t')$. In this section, we only consider the steady state, unless otherwise stated. We assume that the Langevin equation~\eqref{eq:general} has a stable steady distribution, for which the following condition is necessary~\cite{Gardiner}: 
\balign{
a^{11}+a^{22}<0, \  a^{11}a^{22}-a^{12}a^{21}>0 \label{eq:condition}.
}

From the direct calculation from \eref{eq:general}, the analytical expressions of $\sigma_{X},\ \sigma_{Y},\ \sigma_{\rm tot},\ \eta_{Y},\ \overline{C}_Y$ have been obtained~\cite{Hartich2016} (see Appendix A for explicit formulas).  We note that these quantities  are functions of the noise intensity ratio $r:=T^X/T^Y$, rather than individual noise intensities. 

Based on the analytical expressions, we obtain the trade-off relationship between the sensory capacities and the information-thermodynamic efficiency in a simple form. In fact, it has been shown in \cite{Hartich2016} that $C_Y$ and $\eta_Y$ satisfy
\balign{ 
4\eta_{Y}(1-\eta_{Y})\leq C_Y \leq 2\sqrt{\eta_{Y}(1-\eta_{Y})}, \label{eq:con}
}
Furthermore, if we consider the s-SC, the above inequalities reduce to a single equality, which has been shown in \cite{phd}:
\balign{
\overline{C}_Y=4\eta_{Y}(1-\eta_{Y}). \label{eq:claim4}
} 
These relations illuminate that the informational quantities and the entropy production have a nontrivial connection.

We now consider optimization of $\eta_Y$, $\sigma_{\rm tot}$, $\overline{C}_Y$ with respect to $r$, where optimization means maximization for $\eta_Y$ and  $\overline{C}_Y$, but minimization for  $\sigma_{\rm tot}$.
Such optimization can be analytically performed just by differentiating these quantities by $r$, by noting conditions $r>0$ and \eqref{eq:condition}.  

To avoid too much complication, we assume that $a_{11} = 0$ in the following.  This condition is indeed satisfied in a model of \textit{E. Coli} chemotaxis, which is discussed in the next section.
In this case, we obtain the following optimized quantities (see  Appendix A for details). The optimal information-thermodynamic efficiency is given by
\balign{
\eta^{\ast}_{Y}=\frac{(a^{22})^2}{-4a^{12}a^{21}+(a^{22})^2}, \ {\rm where}\ \ \ r^{\ast}=-\frac{a^{12}}{a^{21}}.
}
and the optimal total entropy production is given by
\balign{
\sigma^{\ast}_{\rm tot}&=\frac{4a^{12}a^{21}}{a^{22}}, \ {\rm where}\ \ \ r^{\ast}=-\frac{a^{12}}{a^{21}},
}
For the s-SC, we have the following two cases. When $4a^{12}a^{21}+(a^{22})^2>0$, the optimal s-SC is given by
\balign{
\overline{C}^{\ast}_Y=1,
}
where
\balign{
r^{\ast}=\frac{2a^{12}a^{21}+a^{22}(a^{22}\pm\sqrt{4a^{12}a^{21}+(a^{22})^2})}{2(a^{21})^2}.
}
Here, Eq.~(43) suggests that $y_t$ is a sufficient statistic of $x_t$ in the sense of (36), given that we suppose the steady state here. When $4a^{12}a^{21}+(a^{22})^2<0$, on the other hand, we have
\balign{
\overline{C}^{\ast}_Y=-\frac{16a^{12}a^{21}(a^{22})^2}{(-4a^{12}a^{21}+(a^{22})^2)^2}, \ {\rm where}\ \ \ r^{\ast}=-\frac{a^{12}}{a^{21}}.
}

From the above results, we find that the optimal noise intensity ratio $r^{\ast}$ is the same for $\eta_{Y}$, $\sigma_{\rm tot}$, and also for $\overline{C}_Y$ of the second cese.
This coincidence of optimization would be surprising, though it is just a specific characteristic of linear Langevin systems. In addition, we can show that $\sigma_{X} = \sigma_{Y}$ holds with the same noise intensity ratio $r^\ast$.
\section{\label{sec:BC}Sensory Adaptation}
Investigation of sensory adaptation in terms of stochastic thermodynamics has recently attracted attention~\cite{Barato2013,Itoecoli, Hartich2016, Sartori2014}. Here, we specifically consider signal transduction in {\it E. coli} chemotaxis as an example described by a linear Langevin equation. In particular, we investigate the s-SC and the information-thermodynamic efficiency.
\begin{figure}[b]
\includegraphics[scale=0.4]{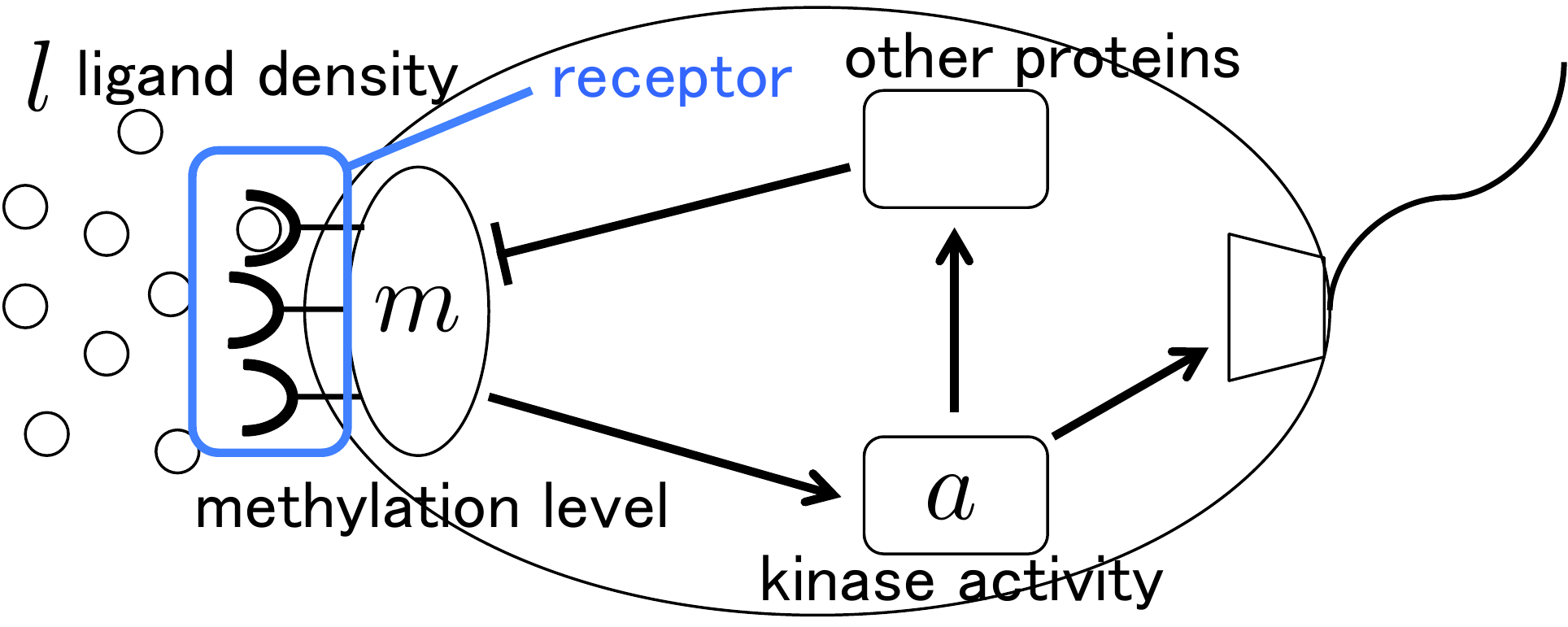}
\caption{\label{fig:chemo}Schematic of signal transduction of {\it E. coli} chemotaxis. The receptor senses the external ligand concentration $l$ and changes the kinase activity $a$, which affects  the methylation level $m$ of the receptor. The methylation level again affects the kinase activity, which makes a negative feedback loop.}
\end{figure}

\subsection{Model of chemotaxis}
The mechanism of the signal transduction of \textit{E. Coli} chemotaxis can be illustrated as follows (see also \fref{chemo})~\cite{Barkai1997}. First, ligands outside a cell are sensed by a receptor on the cell surface, and then change the degree of kinase activation in the cell. Then, the degree of kinase activation changes the methylation level of the receptor through another protein inside the cell. The methylation level  then affects the kinase activity, which makes a negative feedback loop between them.  

The above process can be modeled by the following Langevin equation with the ligand signal $l$, kinase activity $a$, and the methylation level $m$~\cite{Itoecoli,Lan2012,Tu2008,Barkai1997}:
\begin{eqnarray}
\begin{cases}
\dot{m}_t=-\frac{1}{\tau_M}a_t+\xi^M_t, \\
\dot{a}_t=-\frac{1}{\tau_A}\left[ a_t-\overline{a}_t(m_t,l_t) \right]+\xi^A_t, \label{eq:ecoli}
\end{cases}
\end{eqnarray}
where $\xi^i_t\ (i=M, A)$ satisfy $\average{\xi^i_t}=0, \average{\xi^i_t\xi^j_{t'}}=2T^i\delta_{ij}\delta(t-t')$, and $\tau_M,\tau_A$ are time constants satisfying $\tau_M\gg\tau_A>0$. Here, $M$ and $A$ represent the methylation level and the kinase activity, respectively. We note that $l$ describes the logarithm of the ligand concentration $l=\log([L]/[L_0])$, where $[L]$ is the original ligand density and $[L_0]$ is some reference concentration value~\cite{Tu2008}.

In \eref{eq:ecoli}, $\overline{a}_t(m_t,l_t)$ is the steady-state value of the kinase activity under $m_t$ and $l_t$ being fixed. With the linear noise approximation, we have a simple form $\overline{a}_t(m_t,l_t)\simeq\alpha m_t-\beta l_t$~\cite{Tostevin2009,Tu2008}. Under this approximation, \eref{eq:ecoli} is a special case of \eref{eq:general}, where $x_t=m_t,\ y_t=a_t,\ a^{11}=0,\ a^{12}=-1/\tau_M,\ a^{21}=-1/\tau_A,\ a^{22}=-\alpha/\tau_A,\ b^{1}=0, b^{2}=-\beta l_t/\tau_A$. We note that in order to realize the perfect adaptation, it is necessary that the dynamics of the methylation level does not depend on its own state~\cite{Tu2008}, which corresponds to the condition $a^{11}=0$ discussed in Sec.~\ref{sec:LLS}. 

In the steady state, the learning rate of the methylation level is negative ($l_M<0$), and that of the kinase activity is positive ($l_A>0$), which implies that information flows from the methylation level to the kinase activity. This is because  the time scale separation $\tau_M \gg \tau_A$~\cite{Lan2016}, which makes the kinase activity follow the methylation level. In this sense, we regard the kinase activity as a memory, which corresponds to $Y$ in the foregoing sections. Correspondingly, we consider the sensory capacity of the kinase activity, $\overline{C}_A := l_A / \overline{T}_{M \to A}$, in the following. 

We note that the model \eqref{eq:ecoli} is a coarse-grained model by neglecting a lot of degrees of freedoms of a real biological system. Therefore, $\sigma_A$ quantifies effective dissipation that is relevant only to the coarse-grained variables (i.e., $m$ and $a$), which gives a lower bound of real dissipation~\cite{Kawaguchi2013}.

\subsection{Steady-state analysis}

\begin{figure*}[htbp]
\includegraphics[scale=0.55]{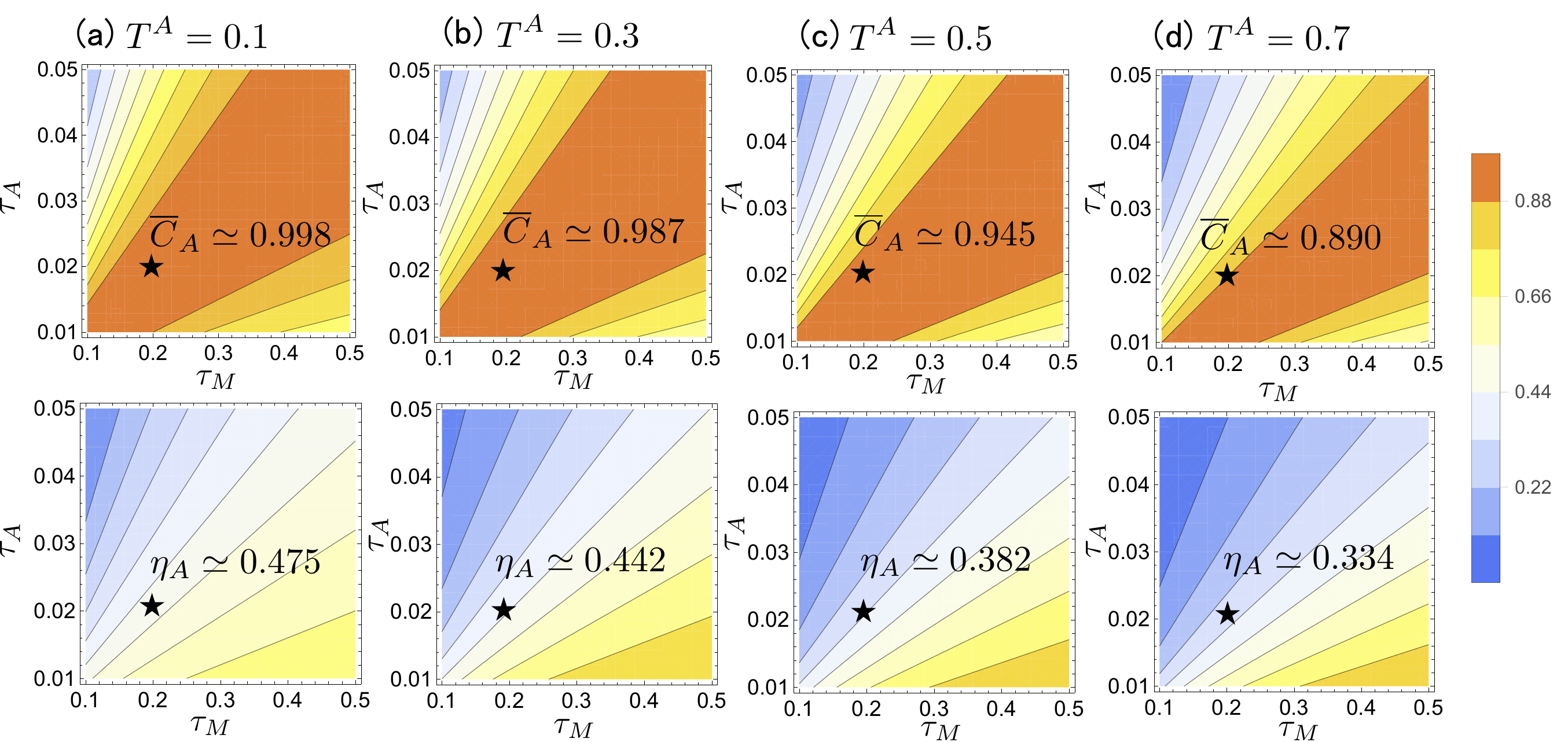}
\caption{\label{fig:capy}Contour plot of the s-SC $\overline{C}_A$ (upper panel) and the information-thermodynamic efficiency $\eta_A$ (lower panel) of the kinase activity. The horizontal and vertical axes are the time constants $\tau_A$ and $\tau_M$, respectively. The noise intensity $T^A$ is chosen to be (a) $T^A=0.1$,\ (b) $T^A=0.3$,\ (c) $T^A=0.5$,\ (d) $T^A=0.7$. The star marks represent the real experimental values of the time constants in {\it E. coli} ($\tau_M=0.2, \tau_A=0.02$)~\cite{Tu2008,Emonet2008,Tostevin2009,Lan2012}. Other parameters are given by $T^M=0.005$ and $\alpha=2.7$, which are consistent with real experiments~\cite{Tu2008,Emonet2008,Tostevin2009,Lan2012}.}
\end{figure*}

\begin{figure*}[htbp]
\includegraphics[scale=0.7]{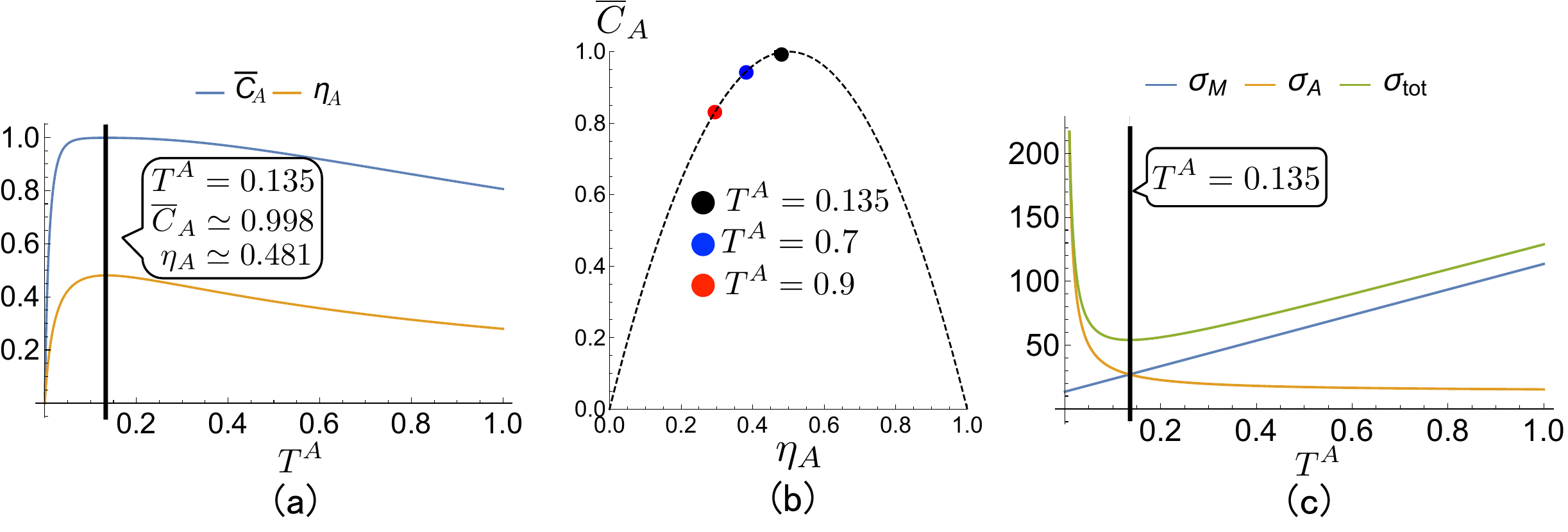}
\caption{\label{fig:graph2}(a) The s-SC $\overline{C}_A$ (blue line) and the information-thermodynamic efficiency $\eta_A$ (orange line) against the noise intensity $T^A$. The optimal noise intensity, at which $\overline{C}_A$ and $\eta_A$ take the maximum value, is indicated by the black bold line. (b) The dotted line represents the trade-off relation \eqref{eq:claim4} between $\overline{C}_A$ and $\eta_A$ in the present setup. The cases of $T^A=0.135$ (black), $T^A=0.7$ (blue), $T^A=0.9$ (red) are indicated by the colored circles, where $T^A=0.135$ (black) is the optimal noise intensity. (c) The entropy productions of the subsystems ($\sigma_M$ and $\sigma_A$) and that of the total system ($\sigma_{\rm tot}$). The noise intensity at which the total entropy production takes the minimum value is indicated by the black bold line. The other parameters are the same as in Fig.~\ref{fig:capy}. } 
\end{figure*}

We first consider the steady state of the chemotaxis model with realistic experimental parameters~\cite{Tu2008,Emonet2008,Tostevin2009,Lan2012}. Figure 3 shows the contour plot of $\overline{C}_A$ and $\eta_A$ as functions of the two time constants. The ranges of the time constants are around the real experimental values that are marked by stars. We note that  $\tau_M\gg \tau_A$ is satisfied here. On the other hand, since the noise intensity $T^A$ of the kinase activity generally depends on fluctuations of the ligand, the precise value of $T^A$ is difficult to access experimentally. Therefore, in \fref{capy}, we have shown several values of $T^A$, which would be consistent with real situations: (a) $T^A=0.1$\ (b) $T^A=0.3$\ (c) $T^A=0.5$\ (d) $T^A=0.7$. We note that condition \eqref{eq:condition} is always satisfied with our choice of parameters in this section.

In all the cases shown in \fref{capy}, we see that the s-SC is almost optimized (i.e., $\overline{C}_A \simeq 1$), while the information-thermodynamic efficiency is far from optimal (i.e., $\eta_A < 0.5)$ as a consequence of the trade-off relation \eqref{eq:claim4}. This suggests that {\it E. Coli} chemotaxis approximately realizes a sufficient statistic at the cost of energetic dissipation.

In addition, we plot $\overline{C}_A$ and $\eta_A$ as functions of  the noise intensity $T^A$ (Fig.~4 (a)). As $T^A$ increases, both of these quantities rapidly increase from zero, take the maximum at $T^A=0.135$ at the same time (as a consequence of the general property of linear Langevin systems as discussed in Sec.~IV), and then gradually decrease.

The maximum value of the s-SC is given by $\overline{C}_A \simeq 0.998$, which is close to unity.
Furthermore, if $T^A$ is not too small (i.e., larger than around $0.05$) and not too large (i.e.,  smaller than around $0.4$), the s-SC $\overline{C}_A$ is close to unity, which is consistent with the results in Fig.~3. 

The maximum value of the information-thermodynamic efficiency is given by $\eta_A \simeq 0.481$.  While this value is far from the maximum $\eta_A =1$ imposed by the second law \eqref{eq:t2nd}, it is sill close to the maximum $\eta_A = 1/2$ under the condition of $\overline{C}_A=1$ (see inequality \eqref{eq:half}). Figure 4 (b) shows the trade-off relation \eqref{eq:claim4} between $\overline{C}_A$ and $\eta_A$, which is a general property of linear Langevin systems.  
From this, we can directly see that $\overline{C}_A \simeq 1$ implies $\eta_A \simeq 1/2$.

In addition, Fig.~4 (c) shows the entropy productions in the same setup. Here, $\sigma_M$ monotonically increases as the noise intensity increases, while $\sigma_A$ monotonically decreases. On the other hand, the total entropy production $\sigma_{\rm tot}=\sigma_M+\sigma_A$ has the minimum at $T^A = 0.135$. This optimal noise intensity coincides with that for $\overline{C}_A$ and $\eta_A$, which is again a consequence of the general argument for linear Langevin systems in the previous section.

\section{Conclusion}\label{sec:conc}
We have investigated the role of sufficient statistics in thermodynamics of information for general CBNs. In particular, we have proved that the presence of a sufficient statistic implies the maximum sensory capacity for both of the multi-time step \eqref{eq:mtrstr} and the single-time step \eqref{eq:mtrstr2} cases.  Given that the sensory capacity is quantitatively related to the information-thermodynamic efficiency as in Eqs.~\eqref{eq:half}, \eqref{eq:con}, \eqref{eq:claim4}, our result leads to constraints on the information-thermodynamic efficiency in the presence of a sufficient statistic.  

In addition, by focusing on linear Langevin systems, we found that in a particular parameter region, there exists optimal noise intensity, with which the information-thermodynamic efficiency, the total entropy production, and the s-SC are optimized at the same time.  As a special case of linear Langevin systems, we have studied a model of {\it E. Coli} chemotaxis.  We then observed that the s-SC takes a nearly maximum value with experimentally realistic parameters, while the information-thermodynamic efficiency is not very high because of the trade-off relation \eqref{eq:claim4}.  

Our result suggests that thermodynamics of information would be useful to investigate biological information processing.  Application of our approach to a broader class of stochastic biological phenomena is a future issue.

\begin{acknowledgments}
We thank Sosuke Ito, Henrik Sandberg and Martin-Luc Rosinberg for useful discussions. T. M. is supported by Grant-in-Aid for JSPS Fellows No. JP17J05540. T. S. is supported by JSPS KAKENHI Grant No. JP16H02211 and No. 25103003.  

\end{acknowledgments}

\appendix

\section{Analytical formulas for the steady state}\label{sec:Sm}
\begin{widetext}
We show the analytical formulas of the informational and thermodynamic quantities for the steady state, which have been derived in Ref.~\cite{Hartich2016} and are used in Sec.~IV and Sec.~V A of this paper. The steady state values of the quantities discussed in Sec.~\ref{sec:LLS} are given as follows:
\balign{
\sigma_{X}&=\frac{a^{12}(a^{21} r-a^{12})}{(a^{11}+a^{22})r}, \label{eq:sigmaxr}\\
\sigma_{Y}&=\frac{a^{21}(-a^{21}r+a^{12})}{a^{11}+a^{22}}, \label{eq:sigmayr}\\
\sigma_{\rm tot}&=-\frac{(a^{21}r-a^{12})^2}{(a^{11}+a^{22})r}, \label{eq:sigmatr}\\
\str&=-\frac{(a^{21})^2((a^{21})^2r^2+(-2a^{12}a^{21}+(a^{11}+a^{22})^2)r+(a^{12})^2)}{4(a^{11}+a^{22})((a^{21})^2r-a^{12}a^{21}+a^{11}(a^{11}+a^{22}))}, \label{eq:strdesu}\\
\eta_{Y}&=\frac{(a^{11}+a^{22})(a^{21}a^{22}r+a^{11}a^{12})}{a^{21}((a^{12})^2-2a^{12}a^{21}r+r((a^{11}+a^{22})^2+(a^{21})^2r))},\label{eq:etayr}\\
\overline{C}_Y&=\frac{4(a^{11}+a^{22})(a^{21}r-a^{12})(a^{21}a^{22}r+a^{11}a^{12})((a^{21})^2 r-a^{12}a^{21}+a^{11}(a^{11}+a^{22}))}{(a^{21})^2((a^{12})^2-2a^{12}a^{21}r+r((a^{11}+a^{22})^2+(a^{21})^2 r))^2}. \label{eq:Cr}
}
\end{widetext}

\section{Analytical formulas for non-steady states}
In this section, we derive analytical expressions of the informational and thermodynamic quantities for non-steady states, which are used to obtain Fig.~5 in Appendix C. 

We consider the following Langevin equation, which is slightly more general than \eref{eq:general}:
\begin{eqnarray}
\begin{cases}
\dot{x}_t=a_t^{11}x_t+a_t^{12}y_t+b_t^1+\xi^X_t, \\
\dot{y}_t=a_t^{21}x_t+a_t^{22}y_t+b_t^2+\xi^Y_t, \label{eq:general2}
\end{cases}
\end{eqnarray}
where $\xi^i_t\ (i=X,Y)$ are the white Gaussian noise terms satisfying $\average{\xi^j_t\xi^j_{t'}}=2T_t^i\delta_{ij}\delta(t-t')$. 
Let $p_t(x,y)$ be the probability density of $(x,y)$ at time $t$.  The corresponding Fokker-Planck equation is given by
\balign{
&\del_tp_t(x,y)=-\del_x J^X_t(x,y)-\del_y J^Y_t(x,y), \label{eq:pt}
}
where
\balign{
&J^X_t(x,y):=\left(a_t^{11}x+a_t^{12}y+b^1_t\right)p_t(x,y)-T^X_t\frac{\del}{\del x}p_t(x,y), \label{eq:Jx} \\
&J^Y_t(x,y):=\left(a_t^{21}x+a_t^{22}y+b^2_t\right)p_t(x,y)-T^Y_t\frac{\del}{\del y}p_t(x,y).
}

We define the following matrix and vectors:
\balign{
A_t:=
\begin{pmatrix}
a_t^{11} &a_t^{12} \\
a^{21}_t &a_t^{22}
\end{pmatrix}, \ \ 
\bm{x}_t := \left(
    \begin{array}{c}
      x_t \\
      y_t          
         \end{array}
  \right), \ \ 
  \bm{b}_t :=\left(
    \begin{array}{c}
      b^1_t \\
      b^2_t          
         \end{array}
  \right).
}
We also consider the ensemble average  and the covariance of $\bm x_t$:
\balign{
\average{\bm{x}_t}&:=\int \bm{x}p_t(x,y) dxdy, \\
\Sigma_t&:=\int (\bm{x}-\average{\bm{x}_t})(\bm{x}-\average{\bm{x}_t})^Tp_t(x,y)dxdy, \label{eq:vcov}
}
where the superscript $T$ is transpose.

We assume that the initial distribution of $(x,y)$ is Gaussian. Since the Langevin equation \eqref{eq:general2} is linear, the probability distribution at any time remains Gaussian,  and is written as
\balign{
p_t(x,y)=\frac{1}{2\pi\sqrt{\det\Sigma_t}}\exp\left[ -\frac{1}{2}(\bm{x}-\average{\bm{x}_t})^T\Sigma^{-1}(\bm{x}-\average{\bm{x}_t}) \right], \label{eq:pxy}
}
which leads to the marginal distributions
\balign{
p_t(x)&=\frac{1}{\sqrt{2\pi\Sigma_t^{11}}}\exp\left[ -\frac{(x-\average{x_t})^2}{2\Sigma_t^{11}} \right], \label{eq:px}\\
p_t(y)&=\frac{1}{\sqrt{2\pi\Sigma_t^{22}}}\exp\left[ -\frac{(y-\average{y_t})^2}{2\Sigma_t^{22}} \right]. \label{eq:py}
}

We then consider the time evolution of $\average{\bm{x}_t}$ and $\Sigma_t$.  We first note that
\balign{
\dot{\average{\bm{x}_t}}=\int \bm{x}\del_t p_t(x,y)dxdy,
}
which leads to
\balign{
\dot{\average{\bm{x}_t}}=A_t\average{\bm{x}_t}+\bm{b}_t. \label{eq:ave}
}
The Riccati equation for $\Sigma_t$ is given by
\balign{
\dot{\Sigma}_t=A_t\Sigma_t+\Sigma_tA^T_t+D, \label{eq:ricc}
}
where we defined
\balign{
D:=
\begin{pmatrix}
2T_t^X &0 \\
0 &2T_t^Y
\end{pmatrix}.
}

We now list the analytical expressions, which will be derived later.
First, the s-TE from $X$ to $Y$ and vice versa are given by
\balign{
&\str=\frac{(a_t^{21})^2}{4T_t^Y}\frac{\det \Sigma_t}{\Sigma_t^{22}}, \label{eq:trxyt}\\
&\overline{T}_{Y\to X}=\frac{(a_t^{12})^2}{4T_t^X}\frac{\det \Sigma_t}{\Sigma_t^{11}}. \label{eq:tryxt}
}
Second, the learning rates are given by
\balign{
&l_{X}=a_t^{12}\frac{\Sigma^{12}_t}{\Sigma^{11}_t}-T_t^X\frac{\Sigma_t^{22}}{\det\Sigma_t}+\frac{T_t^X}{\Sigma_t^{11}}, \label{eq:aaa}\\
&l_{Y}=a_t^{21}\frac{\Sigma^{12}_t}{\Sigma^{22}_t}-T_t^Y\frac{\Sigma_t^{11}}{\det\Sigma_t}+\frac{T_t^Y} {\Sigma_t^{11}}. \label{eq:bbb}
}
Finally, the entropy productions of the subsystems are given by \par
\balign{
\sigma_{X}=&a^{11}_t+a_t^{12}\frac{\Sigma^{12}_t}{\Sigma^{11}_t}+\frac{T^X_t}{\Sigma^{11}_t}+\frac{(a_t^{11})^2}{T_t^X}\Sigma^{11}_t+\frac{(a_t^{12})^2}{T_t^X}\Sigma_t^{22} \notag\\
&+\frac{2a_t^{11}a_t^{12}}{T_t^X}\Sigma^{12}_t+\frac{(a_t^{11})^2}{T_t^X}\average{x_t}^2+\frac{(a_t^{12})^2}{T_t^X}\average{y_t}^2 \notag\\
&+\frac{2a_t^{11}a_t^{12}}{T_t^X}\average{x_t}\average{y_t}+\frac{2a_t^{11}b_t^1}{T_t^X}\average{x_t}+\frac{2a_t^{12}b_t^1}{T_t^X}\average{y_t}\notag \\
&+\frac{(b^1_t)^2}{T_t^X}+a_t^{11}, \label{eq:ccc}\\
\sigma_{Y}=&a^{22}_t+a_t^{21}\frac{\Sigma^{12}_t}{\Sigma^{22}_t}+\frac{T^Y_t}{\Sigma^{22}_t}+\frac{(a_t^{22})^2}{T_t^Y}\Sigma^{22}_t+\frac{(a_t^{21})^2}{T_t^Y}\Sigma_t^{11}\notag \\
&+\frac{2a_t^{22}a_t^{21}}{T_t^Y}\Sigma^{12}_t+\frac{(a_t^{22})^2}{T_t^Y}\average{y_t}^2+\frac{(a_t^{21})^2}{T_t^Y}\average{x_t}^2 \notag\\
&\ \ \ +\frac{2a_t^{22}a_t^{21}}{T_t^Y}\average{y_t}\average{x_t}+\frac{2a_t^{22}b_t^2}{T_t^Y}\average{y_t}+\frac{2a_t^{21}b_t^2}{T_t^Y}\average{x_t}\notag \\
&+\frac{(b^2_t)^2}{T_t^Y}+a_t^{22}. \label{eq:ddd}
}
Among them, the s-TE has already been derived in Ref.~\cite{Ito2016}. Therefore, we only derive the learning rates \eqref{eq:aaa}, \eqref{eq:bbb} and the entropy productions \eqref{eq:ccc}, \eqref{eq:ddd}. 

The general expression of the learning rates are given by~\cite{Horowitz2014,Rosinberg2016}
\balign{
l_{X}=\int J^X_t\del_x\ln \frac{p_t(x,y)}{p_t(x)p_t(y)}dxdy, \label{eq:lx} \\
l_{Y}=\int J^Y_t\del_y\ln \frac{p_t(x,y)}{p_t(x)p_t(y)}dxdy.
}
By using Eqs.~\eqref{eq:pxy}, \eqref{eq:px}, \eqref{eq:py}, we have
\balign{
\del_x\ln \frac{p_t(x,y)}{p_t(x)p_t(y)}=&-\frac{\Sigma_t^{22}}{\det\Sigma_t}(x-\average{x_t})+\frac{\Sigma_t^{12}}{\det\Sigma_t}(y-\average{y_t}) \notag \\
&+\frac{1}{\Sigma_t^{11}}(x-\average{x_t}).
}
\begin{widetext}
Substituting this into \eref{eq:lx}, we obtain
\balign{
l_{X}&=\int \left[\left(a_t^{11}x+a_t^{12}y+b^1_t\right)p_t(x,y)-T^X_t\del_x p_t(x,y)\right]\left[-\frac{\Sigma_t^{22}(x-\average{x_t})}{\det\Sigma_t}+\frac{\Sigma_t^{12}(y-\average{y_t})}{\det\Sigma_t}+\frac{(x-\average{x_t})}{\Sigma_t^{11}}\right]dxdy \notag \\
&=-\frac{\Sigma_t^{22}}{\det\Sigma_t}\left\{ a_t^{11}\Sigma_t^{11}+a_t^{12}\Sigma_t^{12} \right\}+\frac{\Sigma_t^{12}}{\det\Sigma_t}\left\{ a_t^{11}\Sigma_t^{12}+a_t^{12}\Sigma_t^{22} \right\}-T_t^X\frac{\Sigma_t^{22}}{\det\Sigma_t} \notag \\
&\ \ \ +\frac{1}{\Sigma_t^{11}}\left\{ a_t^{11}\Sigma_t^{11}+a_t^{12}\Sigma_t^{12} \right\}+\frac{T_t^X}{\Sigma_t^{11}} \notag \\
&=a_t^{12}\frac{\Sigma^{12}_t}{\Sigma^{11}_t}-T_t^X\frac{\Sigma_t^{22}}{\det\Sigma_t}+\frac{T_t^X}{\Sigma_t^{11}}, \label{eq:lxt}
}
which is \eref{eq:aaa}.
\end{widetext}
In the same manner, we obtain \eref{eq:bbb}. In the steady state, the above result reduces to a simpler form obtained in Ref.~\cite{Horowitz2014}.

We next consider the entropy productions of the subsystems. In Langevin systems, the entropy production is generally rewritten as \cite{Sekimoto}
\balign{
\sigma_{X}=\dot{S}_{X}-\dot{Q}_{X}, \label{eq:abcd}
}
where $S_{X}:=\langle -\ln p(x)\rangle$ is the Shannon entropy of $X$, and $\dot{Q}_{X}$ is the heat flux defined by
\balign{
\dot{Q}_{X}=\frac{\average{(\xi_t^X-\dot{x}_t)\circ \dot{x}_t}}{T^X_t}. \label{eq:abcd2}
}
Here, $\circ$ represents the Stratonovich product \cite{Gardiner}.

We first compute the Shannon entropy in \eref{eq:abcd}. It is straightforward to see that
\balign{
S_{X}=\frac{1}{2}\ln 2\pi \Sigma_t^{11}, \label{eq:S}
}
which leads to
\balign{
\dot{S}_{X}=\frac{1}{2}\frac{\dot{\Sigma}^{11}_t}{\Sigma^{11}_t}.
}
By considering the $(1,1)$ component of \eref{eq:ricc}, we have
\balign{
\frac{\dot{\Sigma}^{11}_t}{\Sigma^{11}_t}=2a_t^{11}+2a^{12}_t\frac{\Sigma^{12}_t}{\Sigma^{11}_t}+\frac{2T_t^X}{\Sigma^{11}_t}.
}
Combining this to \eref{eq:S}, we obtain 
\balign{
\dot{S}_{X}=a^{11}_t+a_t^{12}\frac{\Sigma^{12}_t}{\Sigma^{11}_t}+\frac{T^X_t}{\Sigma^{11}_t}. \label{eq:Sdot}
}

\begin{widetext}

We next consider the heat term \eqref{eq:abcd2}. By transforming the Stratonovich product in \eref{eq:abcd2} to the It$\hat{{\rm o}}$ form, we obtain
\balign{
&\average{(\xi_t^X-\dot{x}_t)\circ \dot{x}_t} \notag\\
&=\average{(-a_t^{11}x_t-a_t^{12}y_t-b_t^1)\circ (a_t^{11}x_t+a_t^{12}y_t+b_t^1+\xi_t^X)} \notag \\
&=-(a_t^{11})^2(\Sigma^{11}_t+\average{x_t}^2)-2a^{11}_ta^{12}_t(\Sigma_t^{12}+\average{x_t}\average{y_t})-2a_t^{11}b_t^1\average{x_t}-(a_t^{12})^2(\Sigma^{22}_t+\average{y_t}^2) \notag\\
&\ \ \ -2a_t^{12}b_t^1\average{y_t}-(b^1_t)^2-a_t^{11}T^X_t.
}
Therefore, we obtain
\balign{
\dot{Q}_{X}&=-\frac{(a_t^{11})^2}{T_t^X}\Sigma^{11}_t-\frac{(a_t^{12})^2}{T_t^X}\Sigma_t^{22}-\frac{2a_t^{11}a_t^{12}}{T_t^X}\Sigma^{12}_t-\frac{(a_t^{11})^2}{T_t^X}\average{x_t}^2-\frac{(a_t^{12})^2}{T_t^X}\average{y_t}^2 \notag \\
&\ \ \ -\frac{2a_t^{11}a_t^{12}}{T_t^X}\average{x_t}\average{y_t}-\frac{2a_t^{11}b_t^1}{T_t^X}\average{x_t}-\frac{2a_t^{12}b_t^1}{T_t^X}\average{y_t}-\frac{(b^1_t)^2}{T_t^X}-a_t^{11}. \label{eq:Q}
}

Combining Eq.\eqref{eq:Sdot} and Eq.\eqref{eq:Q}, we obtain \eref{eq:ccc}. In the same manner, we obtain \eref{eq:ddd}.
\end{widetext}

\section{Example of non-stationary dynamics}

\begin{figure*}[htbp]
\begin{center}
\includegraphics[scale=0.5]{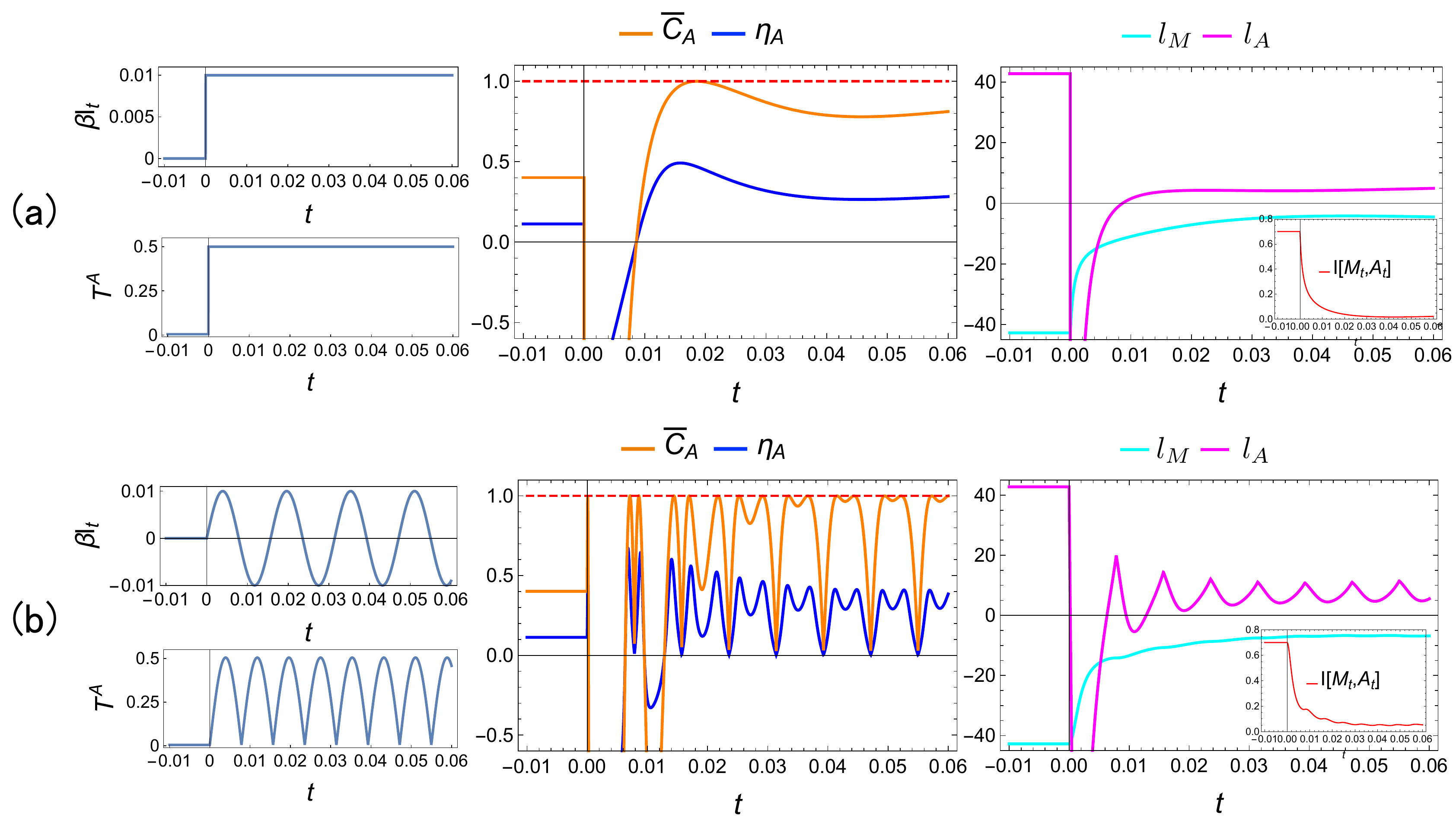} 
\end{center}
\caption{\label{fig:graph}Non-steady dynamics under the time-dependent ligand density and noise (left).  (Middle) The s-SC $\overline{C}_A$ and the information-thermodynamic efficiency $\eta_A$. (Right) The learning rates $l_M$, $l_M$, and the instantaneous mutual information (inset). The initial state is the steady state under $\beta l_t=0$, $T_t^A=0.1$. The external driving of $l_t$ is switched on at $t=0$, where the functional forms at $t>0$ are given by (a) step function: $\beta_t=0.01,\ T_t^A=0.5,$ (b) sinusoidal function: $\beta l_t=0.01\sin(400t),\ T_t^A=0.1+0.4|\sin(400t)|$. The other parameters are the same as in Fig.~\ref{fig:capy}. Since it is difficult to experimentally determine the functional forms of $l_t$ and $T_t^A$, our choice of them is rather arbitrary from the viewpoint of real biology.  The behaviors of $\bar{C}_A$ and $\eta_A$ can easily be changed if we adopt a different choice of the functional form of $T_t^A$.}
\end{figure*}
In this appendix, we consider non-stationary dynamics of the model of {\it E. coli} chemotaxis. As in Sec.~V, we focus on $\overline{C}_A$ and $\eta_A$ of the kinase activity, while 
they can become negative in non-stationary states.

We can analytically compute the s-TE, the learning rate, and the entropy production for the general linear Langevin equation \eqref{eq:general} in non-stationary states. The explicit forms of these quantities and their derivations have already been discussed in Appendix B.

By using these analytical expressions,  we consider non-stationary dynamics of $\overline{C}_A$ and $\eta_A$ for the model of {\it E. Coli} chemotaxis \eqref{eq:ecoli}. Figure 5 shows two types of dynamics of these quantities, where the time evolutions of the external ligand density $l_t$ are given by (a) step and (b) sinusoidal functions. Here, fluctuations of the ligand density is supposed to be incorporated into the noise term $\xi_t^A$, where $l_t$ represents the average of the ligand density.   Since the ligand fluctuation is expected to increase when its average increase, we have adopted the same functional form for of both $l_t$ and $T_t^A$ (here we introduced notation $T_t^A$ for representing the time-dependent noise). We remark, however, that the following results depends on the choice of the functional form of $T_t^A$, and thus should not be regarded as a general feature of {\it E. Coli} chemotaxis from the real biological point of view.

In \fref{graph} (a), both of $\overline{C}_A$ and $\eta_A$ rapidly decrease after $t=0$, and then come close to their optimal values under the constraint of the trade-off relation \eqref{eq:claim4}.

We have also shown the learning rates $l_M$ and $l_A$ (the rightmost panel of Fig.~5). The sum of these learning rates is the time derivative of the instantaneous mutual information (inset). We note that $l_M$ remains negative throughout the dynamics.

The trade-off relation \eqref{eq:claim4} between $\overline{C}_A$ and $\eta_A$ is applicable only to steady states.  However,  in the non-steady dynamics in Fig.~\ref{fig:graph}, $\eta_A$ stays around or below of the half of $\overline{C}_A$.  We can then expect that there might exist a trade-off relation even in non-stationary states, which is an open problem.


\end{document}